\documentclass[11pt,twoside]{article}

\usepackage[hmargin=1.25in,vmargin=1in]{geometry}
\usepackage[utf8]{inputenc}
\usepackage[T1]{fontenc}

\usepackage{jeffe}
\usepackage[charter]{mathdesign}
\usepackage{sourcesanspro,inconsolata}

\usepackage{graphicx}
\usepackage{cite}
\usepackage{soul}
\usepackage{tabularx}

\usepackage{caption}
\newtheorem{theorem}{Theorem}[section]
\newtheorem{lemma}[theorem]{Lemma}

\def\Rev{\textit{rev}}		
\def\Head{\textit{head}}		
\def\Tail{\textit{tail}}		

\def\Torus{\mathbb{T}}		
\let\Univ\widetilde			

\hypersetup{colorlinks=true, urlcolor=Blue, citecolor=Green, linkcolor=Plum, breaklinks, unicode}
\def\EMPH#1{\textcolor{BrickRed}{\emph{#1}}}

\usepackage{mdframed}		

\definecolor{TODOcolor}{cmyk}{0.05,0,0,0}
\definecolor{TODOtxtcolor}{cmyk}{0,1,1,0}
\newproof{TODO}{\color{TODOtxtcolor}\small\sffamily\bfseries TODO}(\color{TODOtxtcolor}\small\sffamily\parindent1.5em)

\surroundwithmdframed[
	linewidth = 1pt,%
	innerleftmargin = 0.5em,%
	innerrightmargin = 0.5em,%
	linecolor = TODOtxtcolor,%
	backgroundcolor = TODOcolor]{TODO}

\begin{document}

\title{Planar and Toroidal Morphs Made Easier}
\author{%
\href{http://jeffe.cs.illinois.edu}{Jeff Erickson}
\qquad\qquad
\href{https://patrickl.in}{Patrick Lin}
\\[1ex]
University of Illinois Urbana-Champaign
}


\date{}
\maketitle

\begin{abstract} \parindent1.5em\noindent
We present simpler algorithms for two closely related morphing problems, both based on the barycentric interpolation paradigm introduced by Floater and Gotsman, which is in turn based on Floater’s asymmetric extension of Tutte’s classical spring-embedding theorem.

First, we give a very simple algorithm to construct piecewise-linear morphs between planar straight-line graphs.  Specifically, given isomorphic straight-line drawings $\Gamma_0$ and $\Gamma_1$ of the same 3-connected planar graph $G$, with the same convex outer face, we construct a morph from $\Gamma_0$ to $\Gamma_1$ that consists of $O(n)$ unidirectional morphing steps, in $O(n^{1+\omega/2})$ time.  Our algorithm entirely avoids the classical edge-collapsing strategy dating back to Cairns; instead, in each morphing step, we interpolate the pair of weights associated with a single edge.

Second, we describe a natural extension of barycentric interpolation to geodesic graphs on the flat torus.  Barycentric interpolation cannot be applied directly in this setting, because the linear systems defining intermediate vertex positions are not necessarily solvable.  We describe a simple scaling strategy that circumvents this issue. Computing the appropriate scaling requires $O(n^{\omega/2})$ time, after which we can can compute the drawing at any point in the morph in $O(n^{\omega/2})$ time.  Our algorithm is considerably simpler than the recent algorithm of Chambers \etal\ and produces more natural morphs.  Our techniques also yield a simple proof of a conjecture of Connelly \etal\ for geodesic torus triangulations.


\end{abstract}

\section{Introduction}

Computing morphs between geometric objects is a fundamental problem that has been well studied, with many applications in graphics, animation, modeling, and more.  A particularly well-studied setting is that of morphing between planar straight-line graphs. Formally, a morph between two isomorphic planar straight-line graphs $\Gamma_0$ and $\Gamma_1$ consists of a continuous family of planar straight-line graphs $\Gamma_t$ starting at $\Gamma_0$ and ending at $\Gamma_1$.

We describe an extremely simple morphing algorithm for planar graphs, which simultaneously obtains properties of two earlier approaches: Floater and Gotsman's barycentric interpolation method~\cite{fg-mti-99,gs-gipm-01,sg-cmcpt-01,sg-msfuo-01,sg-imct-03} results in morphs that are natural and visually appealing but are represented implicitly; variations on Cairns' edge-collapse method~\cite{c-dprc-44,c-idgc2-44,t-dpg-83,aabcd-hmpgd-17,kklss-cimpg-19} result in efficient explicit representations of morphs that are not useful for visualization.  Our new algorithm efficiently computes an explicit piecewise-linear representation of a morph between drawings of the same $3$-connected planar graph, that are potentially more useful for visualization than morphs based on Cairns' method.

We also extend Floater and Gotsman's planar morphing algorithm to geodesic graphs on the flat torus.  Recent results of Luo, Wu, and Zhu \cite{lwz-dsgtg-21} imply that Floater and Gotsman's method directly generalizes to morphs between geodesic triangulations on surfaces of negative curvature, but a direct generalization to the torus generically fails~\cite{sf-ppc2m-04}.  Our extension is based on simple scaling strategy, and it yields more natural morphs than previous algorithms based on edge collapses~\cite{celp-hmgt-21}.  
 Finally, our arguments yield a straightforward proof of a conjecture of Connelly, Henderson, Ho, and Starbird \cite{chhs-prlhe-83} on the structure of the deformation space of geodesic triangulations of the torus.

\subsection{Related Work}

\subsubsection{Planar Morphs}

The history of morphing arguably begins with Steinitz \cite[p.~347]{sr-vtp-34}, who proved that any 3-dimensional convex polyhedron can be continuously formed into any other convex polyhedron with the same 1-skeleton.

Cairns \cite{c-dprc-44,c-idgc2-44} was the first to prove the existence of morphs between  arbitrary isomorphic planar straight-line triangulations, using an inductive argument based on the idea of collapsing an edge from a low-degree vertex to one of its neighbors.  Thomassen \cite{t-dpg-83} extended Cairns’ proof to arbitrary planar straight-line graphs.  Cairns and Thomassen’s proofs are constructive, but yield morphs consisting of an exponential number of steps.

Floater and Gotsman \cite{fg-mti-99} proposed a more direct method to construct morphs between planar graphs, based on an extension by Floater \cite{f-ptsda-98} of Tutte’s classical spring embedding theorem \cite{t-hdg-63}.  Let $\Gamma$ be a straight-line drawing of a planar graph~$G$, such that the boundary of every face of $\Gamma$ is a strictly convex polygon.  Then every interior vertex in $\Gamma$ is a strict convex combination of its neighbors; that is, we can associate a positive weight $\lambda_{\arc{u}{v}}$ with each half-edge or \emph{dart} $\arc{u}{v}$ in $G$, such that the vertex positions $p_v$ in $\Gamma$ satisfy the linear system
\begin{equation}
	\sum_{\arc{u}{v}}^{\phantom{.}} \lambda_{\arc{u}{v}} (p_v - p_u) = (0,0)
	\qquad \text{for every interior vertex $u$}		
	\label{eq:floater}
\end{equation}
Floater \cite{f-ptsda-98} proved that given arbitrary%
\footnote{Floater’s presentation assumes that $\sum_{\arc{u}{v}} \lambda_{\arc{u}{v}} = 1$ for every interior vertex $v$, but this assumption is clearly unnecessary.}
positive weights $\lambda_{\arc{u}{v}}$ and an arbitrary convex outer face, solving linear system~\eqref{eq:floater} yields a straight-line drawing of~$G$ with convex faces.  Tutte’s original spring-embedding theorem \cite{t-hdg-63} is the special case of this result where every dart has weight $1$, but his proof extends verbatim to arbitrary symmetric weights, where $\lambda_{\arc{u}{v}} = \lambda_{\arc{v}{u}}$ for every edge $uv$~\cite{hk-prga-92,r-rsp-96,t-tst-04}.

Floater and Gotsman \cite{fg-mti-99} construct a morph between two convex drawings of the same planar graph $G$, with the same outer face, by linearly interpolating between weights $\lambda_{\arc{u}{v}}$ consistent with the initial and final drawings.  Appropriate initial and final weights can be computed in $O(n)$ time using, for example, Floater’s mean-value coordinates~\cite{f-mvc-03,hf-mvcap-06}.  The resulting morphs are natural and visually appealing.  However, the motions of the vertices are only computed implicitly; vertex positions at any time can be computed in $O(n^{\omega/2})$ time by solving a linear system via nested dissection~\cite{lrt-gnd-79,ay-msnda-13}, where $\omega < 2.37286$ is the matrix multiplication exponent~\cite{l-ptfmm-14,av-rlmfm-21}.%
\footnote{By solving the underlying linear system symbolically, it is possible to express the motion of each vertex as a rational function of degree $n-1$, but computing vertex coordinates at any particular time from this representation requires $O(n^2)$ time.}
Gotsman and Surazhsky generalized Floater and Gotsman’s technique to arbitrary planar straight-line graphs \cite{sg-msfuo-01,gs-gipm-01,sg-cmcpt-01,sg-imct-03}.

A long series of later works, culminating in a paper by
Alamdari, Angelini, Barrera-Cruz, Chan, Da Lozzo, Di Battista, Frati, Haxell, Lubiw, Patrignani, Roselli, Singla, and Wilkinson~\cite{aabcd-hmpgd-17}, describe an efficient algorithm to construct planar morphs with explicit piecewise-linear vertex trajectories, all ultimately
based on Cairns’ inductive edge-collapsing strategy.
Given any two isomorphic straight-line drawings (with the same rotation system and nesting structure) of the same $n$-vertex planar graph, 
the resulting algorithm constructs a morph consisting of $O(n)$ \emph{unidirectional} morphing steps, in which all vertices move along parallel lines at fixed speeds.  Thus, each vertex moves along a piecewise-linear path of complexity $O(n)$, and the entire morph has complexity $O(n^2)$. 
These results require several delicate arguments; in particular, to perturb the \emph{pseudomorphs} defined by edge collapses and their reversals into true morphs.
Recent results of Klemz~\cite{k-cdhgl-21} imply that this algorithm can be implemented to run in $O(n^2\log n)$ time on an appropriate real RAM model of computation.
The resulting morph contracts all vertices into an exponentially small neighborhood and then expand them again, so it is not useful for visualization.

Angelini, Da Lozzo, Frati, Lubiw, Patrignani, and Roselli
Angelini~\cite{alflp-omcd-15} consider the setting of \emph{convexity-preserving} morphs between convex drawings; 
Kleist, Klemz, Lubiw, Schlipf, Staals, and Strash~\cite{kklss-cimpg-19} consider morphing to convexify any 3-connected planar drawing.  Both describe algorithms that produce piecewise-linear morphs consisting of $O(n)$ steps, and that can be implemented to run in time $O(n^{1+\omega/2})$.  (Klemz~\cite{k-cdhgl-21} conjectures that both running times can be improved to $O(n^2\log n)$.)  Combining these algorithms results in an alternative piecewise-linear morph between 3-connected planar drawings.

\subsubsection{Toroidal Morphs}

Until recently, very little was known about morphing graphs on the torus or other more complex surfaces. 

Tutte’s spring-embedding theorem was generalized to simple triangulations of surfaces with non-positive curvature by Colin de Verdière \cite{c-crgtd-91} and independently by Hass and Scott~\cite{hs-seshm-15}.  Delgado-Friedrichs~\cite{d-eppgc-04}, Lovász \cite{l-dafe-04}, and 
Gortler \etal~\cite{ggt-domam-06} also independently proved an extension of Tutte’s theorem to graphs on the flat torus whose universal covers are simple and 3-connected.
For any toroidal graph and any assignment of positive \emph{symmetric} weights to the darts, solving a linear system similar to~\eqref{eq:floater} yields vertex positions of a geodesic drawing with strictly convex faces \cite{ggt-domam-06,el-tmcdc-20}; see Section \ref{S:defs} for details.
%
%
Thus, if two isotopic geodesic torus graphs $\Gamma_0$ and $\Gamma_1$ can both be described by symmetric dart weights, linearly interpolating those weights yields a morph from $\Gamma_0$ to~$\Gamma_1$~\cite{cpv-tbmai-03}; in light of the authors’ toroidal  Maxwell--Cremona correspondence \cite{el-tmcdc-20}, this can be seen as a natural toroidal analogue of Steinitz’s theorem on morphing convex polyhedra \cite{sr-vtp-34}.

The restriction to symmetric weights is both nontrivial and significant.  In a torus graph with convex faces, every vertex can be described as a convex combination of its neighbors, but not necessarily with symmetric weights.  Moreover, the linear system expressing vertex positions as convex combinations of its neighbors is 
rank-deficient, and therefore is
not solvable in general; see Appendix~\ref{A:no-torus-fg} for an example. Thus, Floater’s asymmetric extension of Tutte’s theorem does \emph{not} directly generalize to the flat torus.

For similar reasons, Floater and Gotsman’s planar morphing algorithm also does not generalize.  Suppose we are given two isotopic geodesic torus graphs~$\Gamma_0$ and $\Gamma_1$, each with dart weights 
that express their vertices as convex combinations of their neighbors. 
Unfortunately, in general, interpolating those weights yields linear systems that have no solution; we give a simple example in Appendix~\ref{A:no-torus-fg}.

Steiner and Fischer \cite{sf-ppc2m-04} modify the system by fixing a single vertex, restoring full rank.  However, solving the modified system does not necessarily yield a crossing-free drawing, because the fixed vertex may not lie in the convex hull of its neighbors.%
\footnote{Steiner and Fischer incorrectly claim \cite[Section~2.2.1]{sf-ppc2m-04} that the resulting drawing has no “foldovers” except at the fixed vertex and its neighbors; see Appendix~\ref{A:no-torus-fg}.}
Moreover, even though the initial and final weights are consistent with crossing-free drawings, averages of those weights may not be.  We give an example of this bad behavior in Appendix~\ref{A:no-torus-sf}.

Chambers, Erickson, Lin, and Parsa~\cite{celp-hmgt-21} described the first algorithm to morph between arbitrary essentially 3-connected geodesic torus graphs.  Their algorithm uses a 
combination of Cairns’ edge-collapsing strategy and spring embeddings to construct a morph consisting of $O(n)$ unidirectional morphing steps, in $O(n^{1+\omega/2})$ time.  
Like planar morphs built from edge collapses, these toroidal morphs contract vertices into small neighborhoods and thus are not suitable for visualization.

Recently, Luo \etal ~\cite{lwz-dsgtg-21} 
generalized Floater’s theorem to geodesic triangulations of arbitrary closed Riemannian 2-manifolds with strictly negative curvature, extending the spring-embedding theorems of Colin de Verdière \cite{c-crgtd-91} and Hass and Scott~\cite{hs-seshm-15} to asymmetric weights.  Their result immediately implies that if two geodesic triangulations of such a surface are homotopic, then linearly interpolating the dart weights yields a continuous family of crossing-free geodesic drawings, or in other words, a morph. 
Their result applies only to surfaces with negative Euler characteristic; alas,
the torus has Euler characteristic $0$.

%
%

\subsection{New Results}

We describe two applications of Floater and Gotsman’s barycentric interpolation strategy, which yield simpler algorithms for morphing planar and toroidal graphs.

First we describe a very simple algorithm to construct piecewise-linear morphs between planar straight-line graphs.  Given two isomorphic planar straight-line graphs $\Gamma_0$ and~$\Gamma_1$ with strictly convex faces and the same outer face, we construct a morph from $\Gamma_0$ to $\Gamma_1$ that consists of $O(n)$ unidirectional morphing steps, in $O(n^{1+\omega/2})$ time.
Our morphing algorithm computes barycentric weights for the darts in $\Gamma_0$ and $\Gamma_1$ in a preprocessing phase, and then for each morphing step, interpolates only the pair of weights associated with a single edge.
%
%
Our key observation is that changing the weights for a single edge $e$ moves all vertices in the Floater drawing along lines parallel to~$e$.  (The same observation was made for \emph{symmetric} edge weights by Chambers \etal~\cite{celp-hmgt-21}.)  Our algorithm is significantly simpler than that of Angelini \etal~\cite{alflp-omcd-15} for computing convexity-preserving morphs.  We then extend our algorithm to drawings with non-convex faces, using a simpler approach than Kleist \etal~\cite{kklss-cimpg-19}.
Fig.~\ref{F:plane-example} shows a morph computed by our algorithm; in each frame, the weights of the red edge are about to change.

\begin{figure}[htb]
\centering
\includegraphics[width=0.19\linewidth]{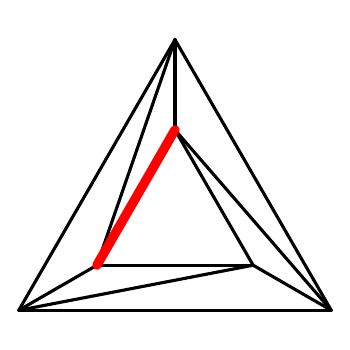}
\includegraphics[width=0.19\linewidth]{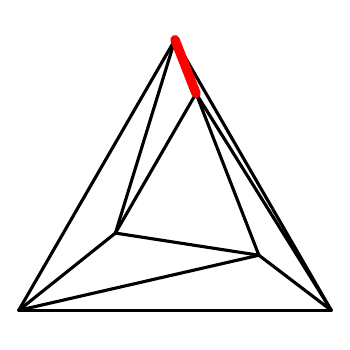}
\includegraphics[width=0.19\linewidth]{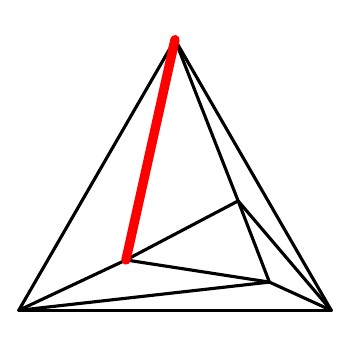}
\includegraphics[width=0.19\linewidth]{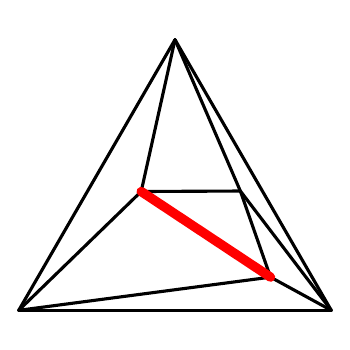}
\includegraphics[width=0.19\linewidth]{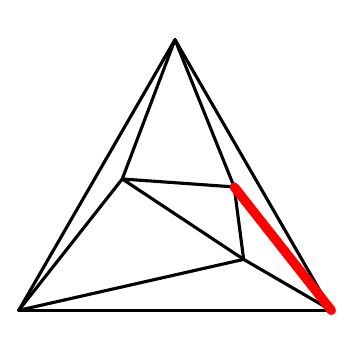}\\
\includegraphics[width=0.19\linewidth]{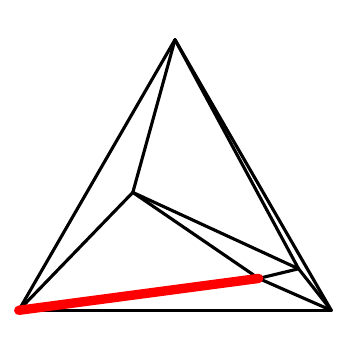}
\includegraphics[width=0.19\linewidth]{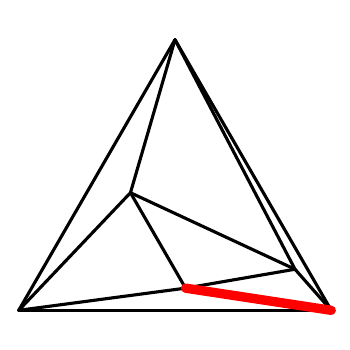}
\includegraphics[width=0.19\linewidth]{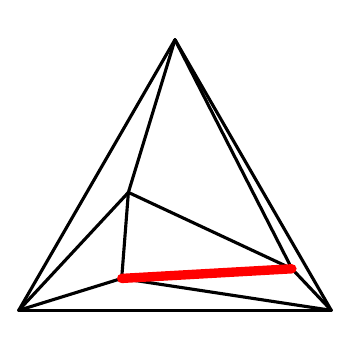}
\includegraphics[width=0.19\linewidth]{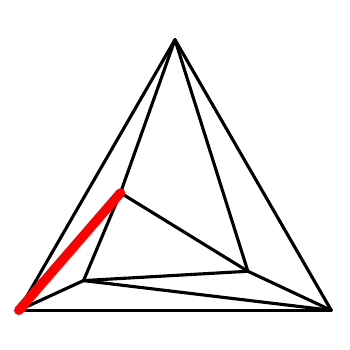}
\includegraphics[width=0.19\linewidth]{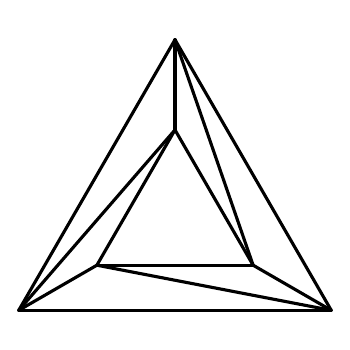}
\caption{Incrementally morphing between planar graphs.}
\label{F:plane-example}
\end{figure}

Next, we describe a natural extension of Floater and Gotsman's method to geodesic graphs on the flat torus.  Our key observation is that barycentric dart weights can be \emph{scaled} so that barycentric interpolation works.  Specifically, we call a weight assignment \emph{morphable} if every \emph{column} of the resulting Laplacian linear system sums to zero; averages of morphable weights are morphable.  Given any weight assignment consistent with any convex drawing, we can guarantee morphability by scaling the weights of all darts leaving each vertex $v$---or equivalently, scaling each \emph{row} of the linear system---by a common positive scalar~$\alpha_v$.  This scaling obviously has no effect on the solution space of the system.  Positivity of the scaling vector $\alpha$ follows from a weighted directed version of the matrix-tree theorem \cite{t-detie-48,b-ueied-60,l-epmtt-20}.  We can computing the appropriate scaling in $O(n^{\omega/2})$ time, after which we can compute any intermediate drawing in $O(n^{\omega/2})$ time, matching the performance of Floater and Gotsman exactly.  The resulting morphs are natural and visually appealing, and our proofs of correctness are considerably simpler than those of Chambers \etal~\cite{celp-hmgt-21}.  However, unlike Chambers \etal, our new morphing algorithm does not compute explicit vertex trajectories.
Fig.~\ref{F:example} shows a morph computed by our algorithm between two randomly shifted $6\times 6$ toroidal grids.  (The authors’ Python implementation is available on request.)

\begin{figure}[htb]
\centering
\includegraphics[width=0.18\linewidth]{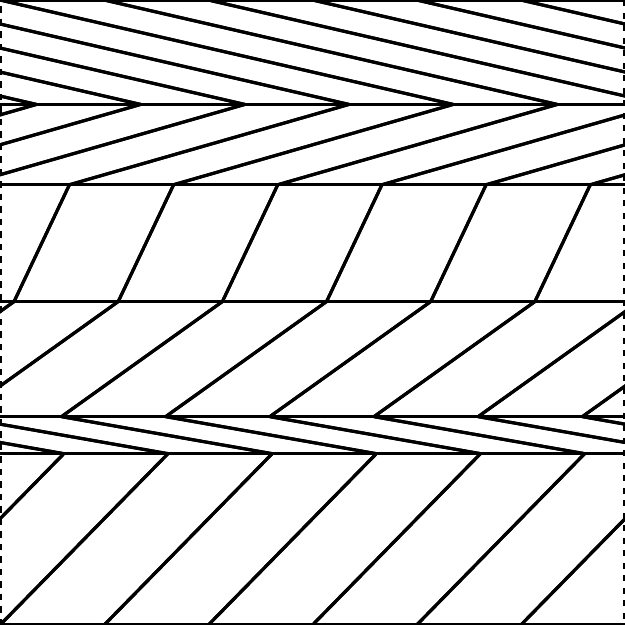}\hfill
\includegraphics[width=0.18\linewidth]{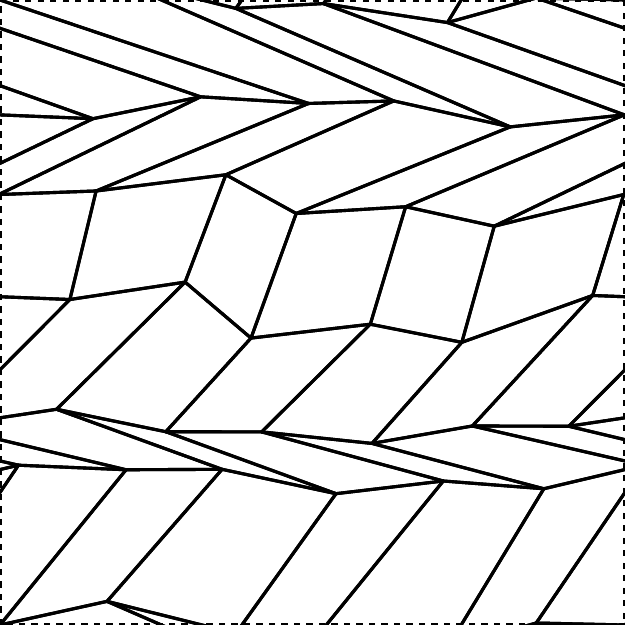}\hfill
\includegraphics[width=0.18\linewidth]{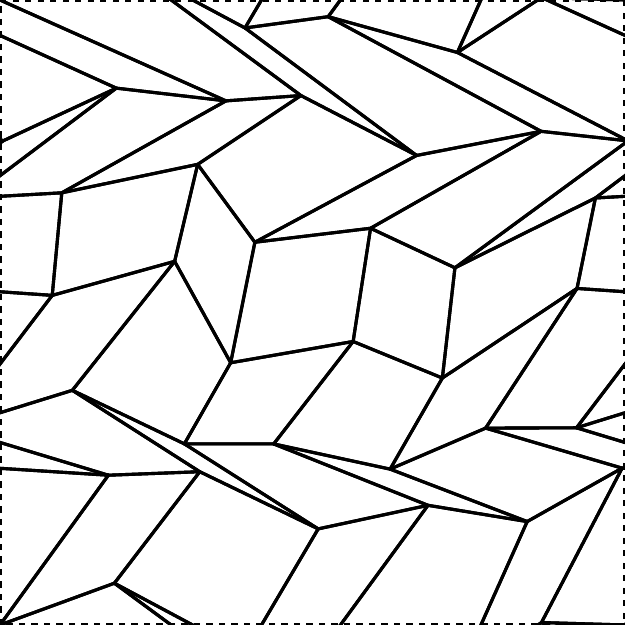}\hfill
\includegraphics[width=0.18\linewidth]{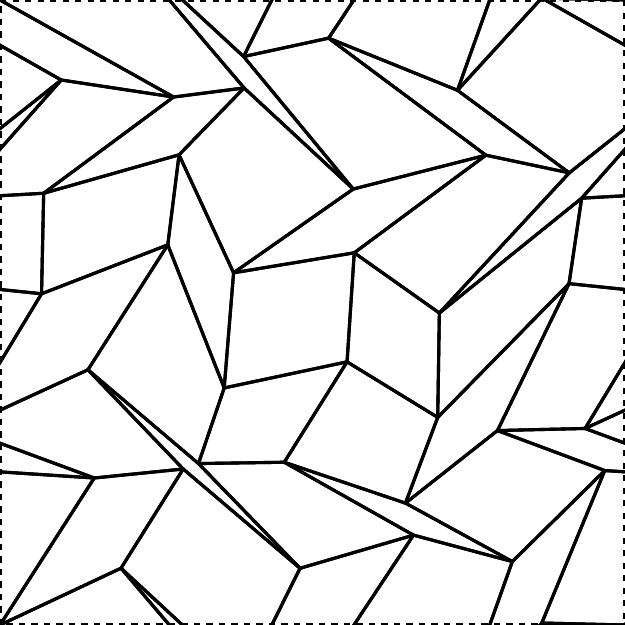}\hfill
\includegraphics[width=0.18\linewidth]{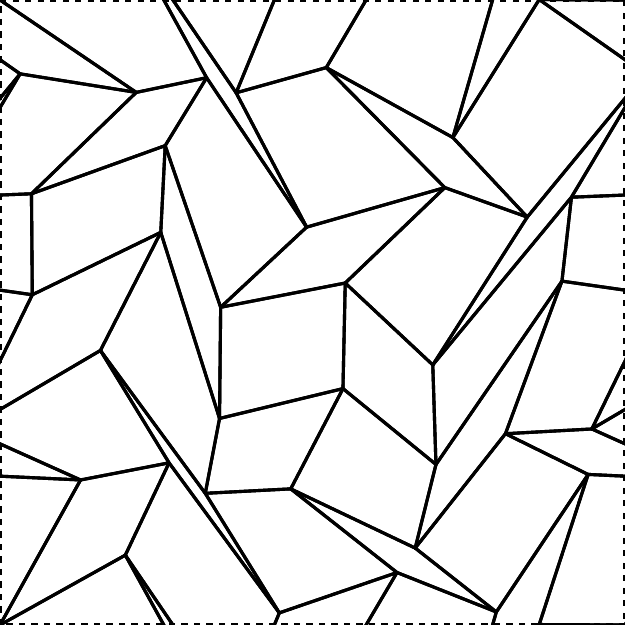}\\[2ex]
\includegraphics[width=0.18\linewidth]{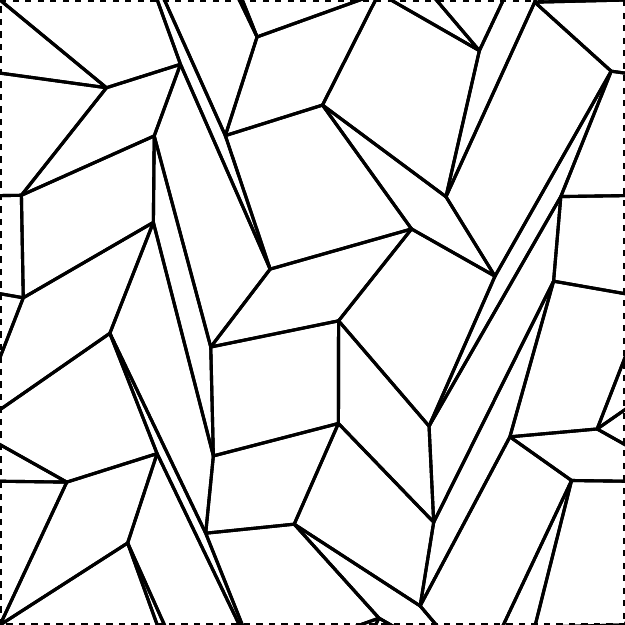}\hfill
\includegraphics[width=0.18\linewidth]{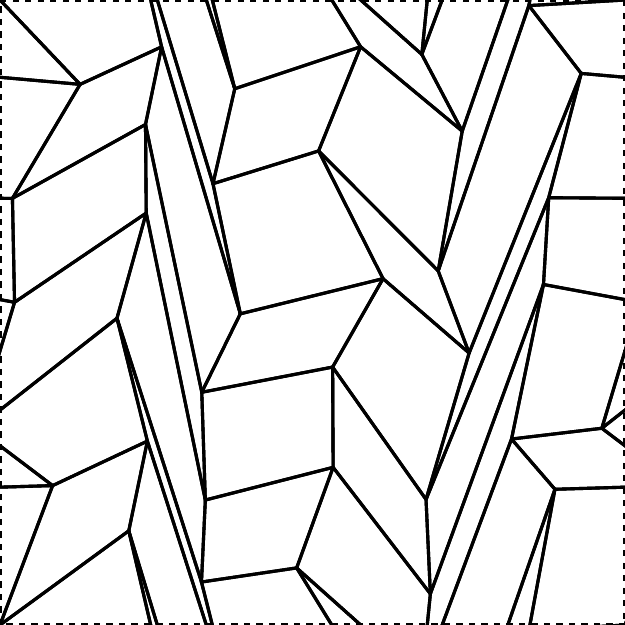}\hfill
\includegraphics[width=0.18\linewidth]{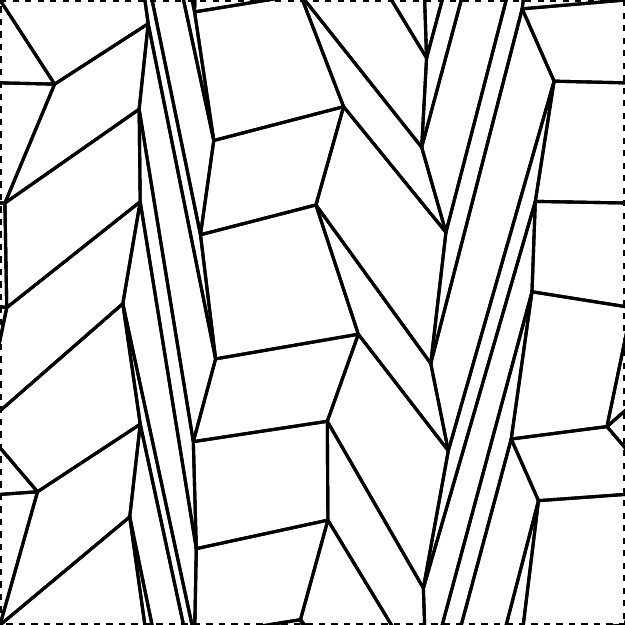}\hfill
\includegraphics[width=0.18\linewidth]{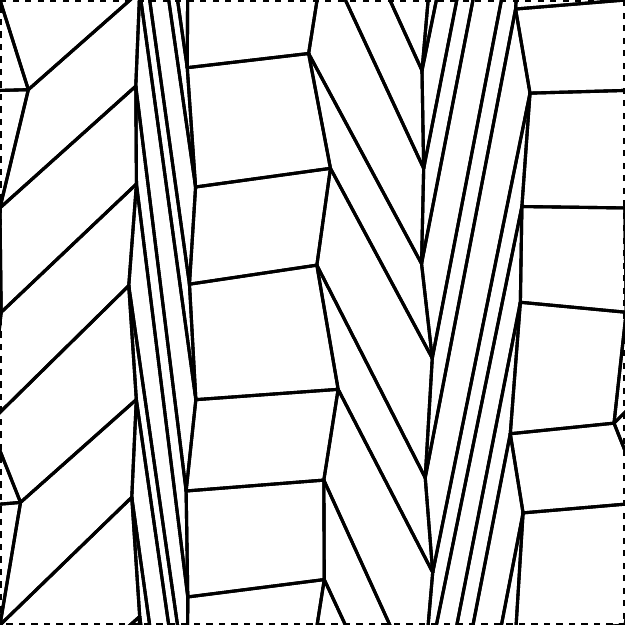}\hfill
\includegraphics[width=0.18\linewidth]{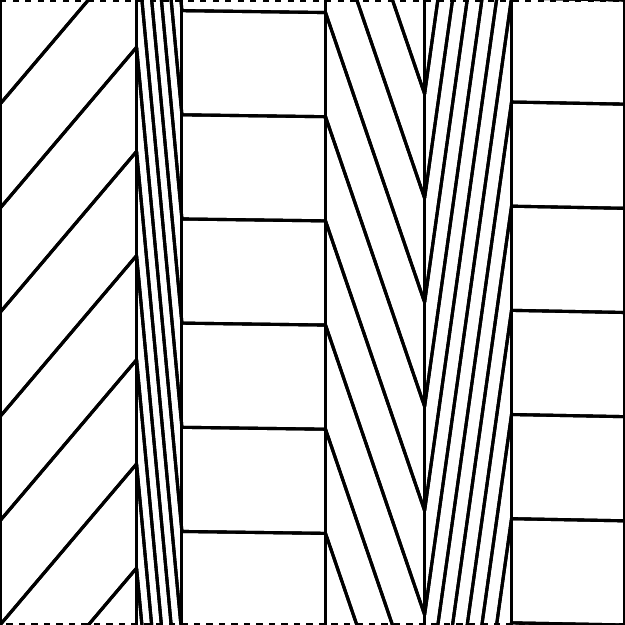}
\caption{Morphing between randomly shifted toroidal grids.}
\label{F:example}
\end{figure}

It remains an open question whether our results can be combined to compute explicit low-complexity piecewise-linear toroidal morphs without edge collapses.  We offer some preliminary observations in Appendix~\ref{A:no-torus-edgebyedge}.

\section{Definitions and Notation}
\label{S:defs}

\subsection{Planar Graphs}

Any planar straight-line drawing $\Gamma$ can be represented by a \EMPH{position} matrix $P \in \Real^{n\times 2}$, each row~$p_v$ of which gives the location of some vertex $v$. Thus, each edge $uv$ is drawn as the straight-line segment $p_up_v$.  We call a planar drawing  \EMPH{convex} if it is crossing-free, every bounded face is a convex polygon, and the outer face is the complement of a convex polygon.

Formally, we regard each edge of any graph as a pair of opposing half-edges or \emph{darts}, each directed from its \emph{tail} to its \emph{head}.  We write \EMPH{$\Rev(d)$} to denote the reversal of any dart $d$.  For simple graphs, we write \EMPH{$\arc{u}{v}$} to denote the dart with tail $u$ and head $v$.  A \EMPH{barycentric weight vector} for~$\Gamma$ assigns a positive real number $\lambda_{\arc{u}{v}}$ to every dart $\arc{u}{v}$ of a graph, so that the vertex positions~$p_v$ satisfy Floater’s linear system \eqref{eq:floater}.  Conversely, for a fixed graph $G$ with a fixed convex outer face, the \EMPH{Floater drawing $\Gamma^\lambda$} of $G$ with respect to a positive weight vector $\lambda$ is the unique drawing whose vertex positions $p_v$ satisfy system \eqref{eq:floater}.

A \EMPH{morph} between two planar drawings $\Gamma_0$ and $\Gamma_1$ is a continuous family of crossing-free drawings $\Gamma_t$ parametrized by time, starting at $\Gamma_0$ and ending at~$\Gamma_1$.  A morph is \emph{linear} if each vertex moves along a straight line at uniform speed, and \emph{piecewise-linear} if it is the concatenation of linear morphs.  Any piecewise-linear morph can be described by a finite sequence of straight-line drawings, or their position matrices.
  A linear morph is \emph{unidirectional} if vertices move along parallel lines.

\subsection{Torus Graphs}

The \EMPH{flat torus} is the quotient space $\Torus = \Real^2 / \Z^2$, also obtained by identifying opposite sides of the unit square $[0,1]^2$.  A \EMPH{geodesic} on the flat torus is the image of a line segment in $\Real^2$ under the projection map $\pi\colon \Real^2 \to \Torus$ where $\pi(x,y) = (x\bmod 1, y\bmod 1)$.

A (crossing-free) \EMPH{geodesic torus drawing} $\Gamma$ of a graph $G$ maps its vertices to distinct points in~$\Torus$ and its edges to simple, interior-disjoint geodesics.  We explicitly consider geodesic drawings of graphs with loops and parallel edges.  We write \EMPH{$d:\arc{u}{v}$} to declare that $d$ is a dart with tail $u$ and head $v$; we emphasize that (unlike in planar setting) there may be more than one such dart. 

Every geodesic torus drawing $\Gamma$ of a graph $G$ is the projection of an infinite, doubly-periodic planar straight-line graph $\Univ\Gamma$, called the \EMPH{universal cover} of $\Gamma$~\cite{celp-hmgt-21}.  We call $\Gamma$ \EMPH{essentially simple} if its universal cover $\Univ\Gamma$ is simple, and \EMPH{essentially 3-connected} if $\Univ\Gamma$ is 3-connected~\cite{m-cpmec-97,m-cpmpt-97}.  Finally, we call $\Gamma$ a \EMPH{convex} drawing if every face of $\Univ\Gamma$ is strictly convex.  Every convex torus drawing is both essentially simple and essentially 3-connected, since every infinite planar graph with strictly convex faces is 3-connected \cite{d-bdpg-03}.

\paragraph{Coordinate representations.}

Following Chambers \etal~\cite{celp-hmgt-21}, we use a \EMPH{coordinate representation} $(P, \tau)$ for geodesic torus drawings that records
\begin{itemize}
\item 
a \EMPH{position} vector $p_v \in \Real^2$ for each vertex $v$, and
\item
a \EMPH{translation} vector $\tau_d \in \Z^2$ for each dart $d$, such that $\tau_{\Rev(d)} = -\tau_d$.
\end{itemize}
These vectors indicate that each dart $d : \arc{u}{v}$ is drawn as the projection of a line segment from~$p_u$ to $p_v+\tau_d$ in the universal cover $\Univ\Gamma$.  In particular, if we normalize all vertex positions to the half-open unit square $[0,1)^2$, then each translation vector $\tau_d$ indicates the number of times $d$ crosses the vertical boundary of the unit square to the right, and the number of times~$d$ crosses the horizontal boundary of the unit square upward.

Two crossing-free drawings of the same graph on the torus are \EMPH{isotopic} if one can be deformed into the other through a continuous family of (not necessarily geodesic) crossing-free drawings; such a deformation is called an \EMPH{isotopy}.  Two crossing-free drawings are isotopic if and only if their coordinate representations can be \EMPH{normalized} so that their translation vectors agree; this condition can be tested in $O(n)$ time~\cite[Theorem A.1]{celp-hmgt-21},~\cite{cm-tgis-14}.  A \EMPH{geodesic isotopy} or \EMPH{morph} is an isotopy in which all intermediate drawings are geodesic.

\paragraph{Barycentric weights.}

In any convex torus drawing $\Gamma$, the position $p_v$ of each vertex $v$ can be expressed as a convex combination of its neighbors, as follows.  We can assign a weight~$\lambda_d > 0$ to each dart $d$ such that any coordinate representation $(P, \tau)$ of $\Gamma$ satisfies the linear system
\begin{equation}
	\sum_{\vphantom{d}v} \sum_{d:\arc{u}{v}} \lambda_d (p_v - p_u + \tau_d) = (0,0)
	\qquad \text{for every vertex $u$.}		
	\label{eq:ggt}
\end{equation}
We can express this linear system in matrix notation as $L^{\lambda}P = H^{\lambda}$, where
\begin{equation}
	\begin{aligned}
	L^{\lambda}_{ij} &=
		\begin{cases}
			\displaystyle \sum_k \Sum_{d :\arc{i}{k}} \lambda_d & \text{if $i = j$} \\[3ex]
			\displaystyle \sum_{d:\arc{i}{j}} -\lambda_d & \text{otherwise}
		\end{cases} &
	\quad\text{and}\quad
	H^{\lambda}_i &= \sum_j \Sum_{d:\arc{i}{j}}  \lambda_d x_d
	\end{aligned}
	\tag{$2’$}
	\label{eq:ggt-laplacian}
\end{equation}
The (unnormalized, asymmetric) Laplacian matrix $L^\lambda$ has rank $n-1$~\cite{sf-ppc2m-04}.  We call any positive weight vector $\lambda$ satisfying system~\eqref{eq:ggt} \EMPH{barycentric} for $\Gamma$.  Barycentric weights for any convex torus drawing can be computed in $O(n)$ time using, for example, Floater’s mean-value coordinates~\cite{f-mvc-03,hf-mvcap-06}.

On the other hand, suppose we fix the graph $G$ and translation vectors $\tau_d$ consistent with an essentially 3-connected (but not necessarily geodesic) drawing of $G$.  Then for any positive weight vector $\lambda$, any solution to linear system \eqref{eq:ggt} gives the vertex positions~$p_v$ of a convex drawing $\Gamma^\lambda$ of~$G$~\cite{ggt-domam-06}.  In this case, we say that the \EMPH{Floater drawing}~$\Gamma^\lambda$ \EMPH{realizes} the weight vector~$\lambda$, and we call the weight vector $\lambda$ \EMPH{realizable} for the graph $G$.  Every realizable weight vector is realized by a two-dimensional family of drawings that differ by translation.

Every \emph{symmetric} positive weight vector (where $\lambda_d = \lambda_{\Rev(d)}$) is realizable: for any assignment of positive weights to the \emph{edges} of $G$, there is a corresponding convex torus drawing~\cite{c-crgtd-91, d-eppgc-04, l-dafe-04, ggt-domam-06, hs-seshm-15}.  
Realizable weights are not necessarily symmetric: there are convex torus drawings with only asymmetric barycentric weights.  Conversely, positive asymmetric weights are not always realizable.

%


\section{Morphing Planar Graphs Edge by Edge}
\def\Tutte#1#2{{#1}^{#2}}
\def\Unit{\mathbf{e}}

We describe a very simple algorithm to morph planar straight-line graphs that combines the benefits of both the Floater and Gotsman approach~\cite{fg-mti-99,gs-gipm-01,sg-cmcpt-01,sg-msfuo-01,sg-imct-03} and the Cairns approach~\cite{c-dprc-44,c-idgc2-44,t-dpg-83,aabcd-hmpgd-17,kklss-cimpg-19}.  Our algorithm constructs a morph consisting of $O(n)$ unidirectional morphing steps, in $O(n^{1+\omega/2})$ time.  Because our morphs do not use edge collapses, they are also potentially good for visualization.

Fix a planar graph $G$ and a convex outer face.  Let $p^\lambda_v$ denote the position of vertex $v$ in the Floater drawing $\Gamma^\lambda$ with respect to weight vector $\lambda$.  The following lemma is a planar asymmetric version of Lemma 5.1 of Chambers \etal~\cite{celp-hmgt-21}. Intuitively, it states that changing the weights of the darts of a single edge $e$ moves each vertex in the Floater drawing along lines parallel to $e$.

\begin{lemma}
\label{L:tutte-parallel-plane}
Let $\lambda$ and $\mu$ be arbitrary positive weight vectors such that $\lambda_{d} \ne \mu_{d}$ or $\lambda_{\Rev(d)} \ne \mu_{\Rev(d)}$ for some dart $d$, but $\lambda_{d'} = \mu_{d'}$ for all darts $d' \notin \Set{d,\Rev(d)}$. For each vertex $w$, the vector $p^\mu_w - p^\lambda_w$ is parallel to the drawing of $d$ in $\Tutte{\Gamma}{\lambda}$.
\end{lemma}

\begin{proof}
Suppose $d$ has tail $u$ and head $v$, and (by rotating the drawing if necessary) that $d$ is drawn parallel to the $x$-axis.  For each vertex $i$, let $y^\lambda_i$ and $y^\mu_i$ be the $y$-coordinates of points $p^\lambda_i$ and~$p^\mu_i$, respectively, so that $y^\lambda_u = y^\lambda_v$.  We need to prove that $y^\lambda_w = y^\mu_w$ for every vertex $w$.

Projecting linear system~\eqref{eq:floater} for $\lambda$ onto the $y$-axis gives us
\begin{equation}
\label{eq:floater-y}
\sum^{\phantom{.}}_{\arc{i}{j}} \lambda^{\phantom{.}}_{\arc{i}{j}}(y^\lambda_j - y^\lambda_i) = 0 \qquad \text{for each vertex $i$}.
\end{equation}
Swapping entries of $\lambda$ with corresponding entries of $\mu$ in the system~\eqref{eq:floater-y} changes at most two constraints, corresponding to the two endpoints $u$ and $v$ of $d$.  Moreover, in each changed constraint, the single changed coefficient is multiplied by $y^\lambda_u - y^\lambda_v = y^\lambda_v - y^\lambda_u = 0$, so the $y^\lambda_i$'s also solve the corresponding system for $\mu$. Since the system~\eqref{eq:floater-y} and its counterpart for $\mu$ each have a unique solution, we conclude that $y^\lambda_w = y^\mu_w$ for every vertex $w$.
\end{proof}

Under the assumptions of Lemma~\ref{L:tutte-parallel-plane}, linearly interpolating vertex positions from $\Gamma^\lambda$ to $\Gamma^\mu$ yields a unidirectional linear morph~\cite[Corollary 7.2]{aabcd-hmpgd-17},~\cite[Lemma~5.2]{celp-hmgt-21}.  It follows that we can morph between isomorphic convex drawings through a sequence of at most $3n-9$ unidirectional linear morphing steps, one for each internal edge, following the algorithm in Fig.~\ref{fig:convex-planar-algo}.  Initial and final barycentric weight vectors can be found in $O(n)$ time using, for example, Floater's mean-value method~\cite{f-mvc-03,hf-mvcap-06}.  
 Each intermediate drawing can be computed in $O(n^{\omega/2})$ time using  nested dissection~\cite{lrt-gnd-79,ay-msnda-13}, for a total running time of $O(n^{1+\omega/2})$.

\begin{figure}[htb]
\centering
\begin{algorithm}
\textul{$\textsc{MorphConvex}(\Gamma_{\text{start}}, \Gamma_{\text{end}})$}: \+
\\	$\lambda \gets$ barycentric weights for $\Gamma_{\text{start}}$
\\	$\mu \gets$ barycentric weights for $\Gamma_{\text{end}}$
\\	$k \gets 0$
\\[0.5ex]
	for each internal edge $e$ \+
\\		$k \gets k + 1$
\\		$d \gets $ a dart of $e$
\\		$\lambda_{d} \gets \mu_{d}$
\\		$\lambda_{\Rev(d)} \gets \mu_{\Rev(d)}$
\\		$\Gamma_k \gets \Gamma^\lambda$\-
\\[0.5ex]
	return $\Gamma_{\text{start}},\Gamma_1,\Gamma_2,\dots,\Gamma_k$ \Comment{$ = \Gamma_{\normalfont\text{end}}$}
\end{algorithm}
\caption{Algorithm for morphing between convex planar drawings.}
\label{fig:convex-planar-algo}
\end{figure}

Because all Floater drawings are convex, Lemma~5.2 of Chambers \etal~\cite{celp-hmgt-21} implies that \textsc{MorphConvex} actually produces a \emph{convexity-preserving} piecewise-linear morph; all faces remain convex throughout the morph.  The existence of convexity-preserving morphs was first proved by Thomassen~\cite{t-dpg-83}; Angelini, Da Lozzo, Frati, Lubiw, Patrignani, and Roselli~\cite{alflp-omcd-15} described a piecewise-linear convexity-preserving morph consisting of $O(n)$ steps.  
Our algorithm is significantly simpler than that of Angelini \etal~\cite{alflp-omcd-15}.

\begin{theorem}
\label{T:planar-convex-morph}
Given any two isomorphic convex planar drawings 
with $n$ vertices and the same convex outer face, we can compute a morph 
between them consisting of at most $3n - 9$ unidirectional linear morphing steps, in $O(n^{1+\omega/2})$ time.
\end{theorem}

The previous algorithm can be extended to morph between non-convex drawings by introducing  intermediate convex drawings, as follows: Add edges to the initial and final drawings to decompose every face into convex polygons, compute barycentric weights for the resulting convex drawing, and then reduce the weights of each added edge (one-by-one) to zero, effectively deleting that edge.  Dropping the added edges yields a piecewise-linear morph from each input drawing to a convex drawing.  Again, each intermediate drawing can be computed in $O(n^{\omega/2})$ time.  Our complete morphing algorithm is shown in Fig.~\ref{fig:planar-algo}.  Our algorithm \textsc{Convexify} is considerably simpler than that of Kleist \etal~\cite{kklss-cimpg-19}; however, unlike Kleist \etal, our algorithm is \emph{not} necessarily convexity-increasing.

\begin{figure}[htb]
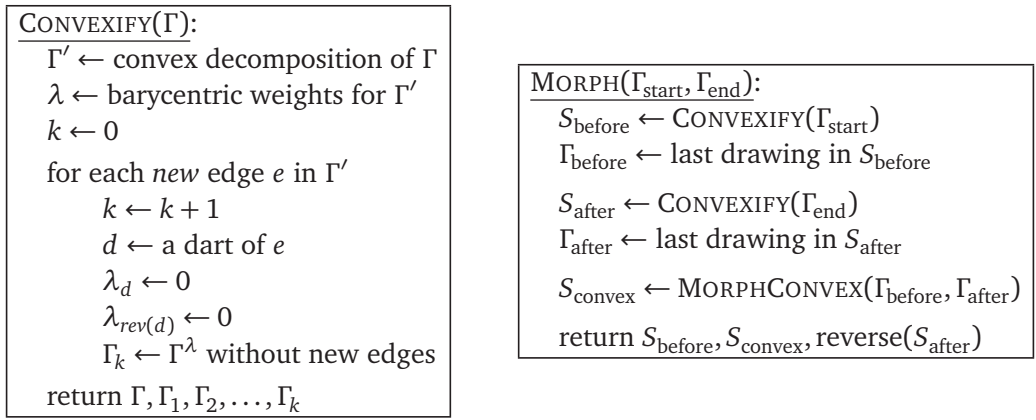

\centering
\begin{algorithm}
\textul{$\textsc{Convexify}(\Gamma)$}: \+
\\  $\Gamma' \gets$ convex decomposition of $\Gamma$
\\  $\lambda \gets$ barycentric weights for $\Gamma'$
\\  $k \gets 0$
\\[0.5ex]
  for each \emph{new} edge $e$ in $\Gamma'$ \+
\\    $k \gets k + 1$
\\    $d \gets $ a dart of $e$
\\    $\lambda_d \gets 0$
\\    $\lambda_{\Rev(d)} \gets 0$
\\    $\Gamma_k \gets \Gamma^\lambda$ without new edges\-
\\[0.5ex]
  return $\Gamma,\Gamma_1,\Gamma_2,\dots,\Gamma_k$
\end{algorithm}
\hfil
\begin{algorithm}
\textul{$\textsc{Morph}(\Gamma_{\text{start}}, \Gamma_{\text{end}})$}: \+
\\  $S_{\text{before}} \gets \textsc{Convexify}(\Gamma_{\text{start}})$
\\  $\Gamma_{\text{before}} \gets $ last drawing in $S_{\text{before}}$
\\[1ex]
    $S_{\text{after}} \gets \textsc{Convexify}(\Gamma_{\text{end}})$
\\  $\Gamma_{\text{after}} \gets $ last drawing in $S_{\text{after}}$
\\[1ex]
	$S_{\text{convex}} \gets \textsc{MorphConvex}(\Gamma_{\text{before}}, \Gamma_{\text{after}})$
\\[1ex]
	return $S_{\text{before}}, S_{\text{convex}}, \text{reverse}(S_{\text{after}})$
\end{algorithm}
\caption{Algorithm for morphing between general planar straight-line drawings.}
\label{fig:planar-algo}
\end{figure}

In total, we perform one morphing step for each internal edge of $G$, plus at most $2(k-3)$ morphing steps for each bounded face with degree $k$.
Euler's formula implies that a $3$-connected planar graph has between $1.5n$ and $3n - 6$ edges, and thus at most $3n - 9$ internal edges.  Thus, we need to add at most $1.5n - 6$ edges to convexify the initial and final faces, so our morph consists of at most $4.5n - 15$ linear morphing steps.
In summary:

\begin{theorem}
\label{T:planar-morph}
Given any two isomorphic $3$-connected planar straight-line drawings 
 with $n$ vertices and the same convex outer face, we can compute a morph 
 between them consisting of at most $4.5n - 15$ unidirectional linear morphing steps, in $O(n^{1+\omega/2})$ time.
\end{theorem}

\paragraph{Models of computation.}
The algorithm of Alamdari \etal~\cite{aabcd-hmpgd-17} require a slightly nonstandard real RAM model of computation that supports exact square and exact cube roots.  In contrast, Floater’s mean-value weights can be expressed in terms of areas and Euclidean lengths~\cite{hf-mvcap-06}, which require only square roots to evaluate exactly.  If initial and final barycentric weights are given, both Floater and Gotsman’s morphing algorithm~\cite{fg-mti-99} and our incremental algorithm use only basic arithmetic operations: addition, subtraction, multiplication, and division.


Even without exact roots, any integer-RAM or floating-point implementation of our morphing algorithm must contend with precision issues.  A careful implementation of Alon and Yuster’s nested dissection algorithm~\cite{ay-msnda-13} solves Floater’s linear system \eqref{eq:floater} exactly in $O(n^{1 + \omega/2}\polylog n)$ bit operations, assuming all dart weights $\lambda_{\arc{u}{v}}$ are $O(\log n)$-bit integers \cite{ay-msnda-13,b-simig-68}.  Thus, at least then all weights are given as part of the input, an exact implementation of our morphing algorithm runs in $O(n^{2+\omega/2}\polylog n)$ on a standard integer RAM.  Coordinates of Tutte/Floater drawings can require $\Omega(n)$ bits of precision to avoid collapsing or crossing edges~\cite{eg-dspgt-95, df-ftfgrp-21}; a canonical bad example is shown in Figure~\ref{F:twisted}.  Thus, the near-linear cost of exact arithmetic is unavoidable in the worst case.
%

Shen, Jiang, Zorin, and Panozzo~\cite{sjzp-pe-19} observe that floating-point implementations of Tutte’s algorithm suffer from robustness issues in practice.  Shen \etal~describe an iterative procedure to repair floating-point-Tutte drawings; however, it is unclear whether a similar procedure can be used to avoid precision issues in our algorithm while maintaining continuity of the resulting morph.  It is also unclear whether precision issues in our algorithm can be avoided, or at least minimized, by carefully choosing the order in which edge weights are changed and/or by morphing through a carefully chosen intermediate drawing.

\begin{figure}[htb]
\centering
\begin{tabular}{c@{\qquad}c@{\qquad}c}
\includegraphics[scale=0.3]{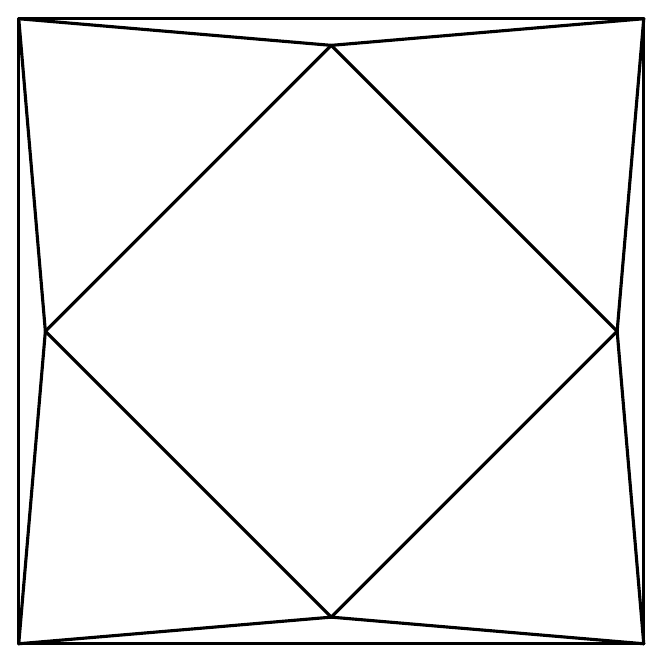} & 
\includegraphics[scale=0.3]{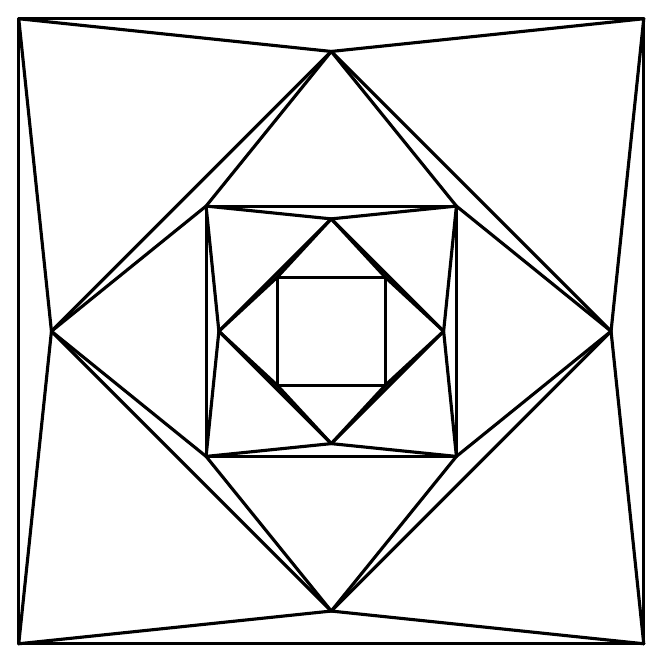} & 
\includegraphics[scale=0.3]{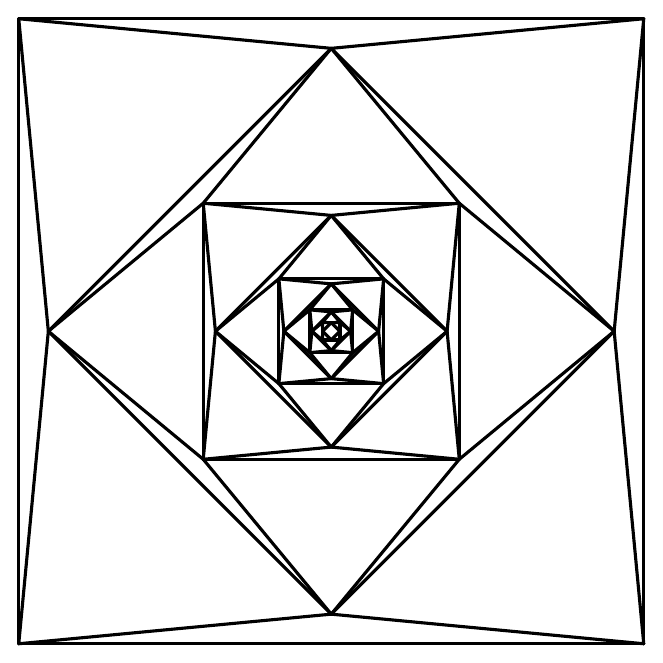}
\end{tabular}
\caption{A family of weighted spring embeddings consisting of nested squares, with 2, 5, and 10 layers, respectively.}
\label{F:twisted}
\end{figure}

\section{Morphable Weight Vectors on the Flat Torus}

As observed by Steiner and Fischer~\cite{sf-ppc2m-04}, Floater and Gotsman’s morphing algorithm does not directly generalize to the toroidal setting, since not all positive  weight vectors $\lambda$ are realizable.  In particular, given arbitrary barycentric weights $\lambda(0)$ and $\lambda(1)$ of two isotopic convex torus drawings, intermediate weights $\lambda(t) := (1-t)\lambda(0) + t\lambda(1)$ are not necessarily realizable; see Appendix~\ref{A:no-torus-fg} for an example.  Thus, interpolating barycentric weights does not necessarily give us a morph.

To bypass this issue, we identify a subspace of \emph{morphable} weight vectors, such that every convex torus drawing has a morphable barycentric weight vector, every morphable weight vector is realizable, and convex combinations of morphable weights are morphable.  Specifically, a positive weight vector $\lambda$ is \EMPH{morphable} if each \emph{column} of the matrices $L^\lambda$ and $H^\lambda$ sums to $0$.%
\footnote{Directed graph Laplacians whose columns sum to zero are also called \emph{Eulerian}; Cohen \etal~\cite{ckpps-facsd-16} refer to the scaling process described by Lemma \ref{L:can-morphabilify} as \emph{Eulerian scaling}.}
The following lemma is immediate:

\begin{lemma}
\label{L:morphable-convex}
Convex combinations of morphable weight vectors are morphable.
\end{lemma}

\begin{lemma}
\label{L:morphable-is-realizable}
Every morphable weight vector is realizable.
\end{lemma}

\begin{proof}
If $\lambda$ is a morphable weight vector, then the $n$th row of the linear system $L^\lambda P = H^\lambda$ is implied by the other $n-1$ rows, so we can remove it. The resulting abbreviated linear system still has rank $n-1$, so it has a (unique) solution. 
\end{proof}

\begin{lemma}
\label{L:can-morphabilify}
Given a barycentric weight vector $\lambda$ for a convex torus drawing~$\Gamma$, a \emph{morphable} barycentric weight vector for $\Gamma$ can be computed in $O(n^{\omega/2})$ time.
\end{lemma}

\begin{proof}
The matrix $L^{\lambda}$ has rank $n-1$, so there is a one-dimensional space of (row) vectors $\alpha = (\alpha_1,\dots,\alpha_n)$ such that $\alpha L^\lambda = (0,\dots,0)$.  We can compute a non-zero vector $\alpha$ in $O(n^{\omega/2})$ time using nested dissection~\cite{lrt-gnd-79,ay-msnda-13,ad-lapeg-96}.

A directed version of the matrix tree theorem \cite{t-detie-48, b-ueied-60, l-epmtt-20} implies that we can choose all $\alpha_i$ to be positive.  Specifically, let $G^\pm$ be the weighted directed graph whose weighted arcs correspond to the weighted darts of $G$.  An \emph{inward directed spanning tree} is an acyclic spanning subgraph of~$G^\pm$ where every vertex except one (called the \emph{root}) has out-degree $1$.  The weight of an inward directed spanning tree is the product of the weights of its arcs.  For each~$i$, let~$\alpha_i$ be the sum of the weights of all inward directed spanning trees rooted at vertex~$i$; we have $\alpha_i>0$ because all dart weights are positive.  The directed matrix tree theorem implies that $\alpha L = 0$, as required; for an elementary proof, see De Leenheer~\cite[Theorem~3]{l-epmtt-20}.  (See also Cohen \etal~\cite[Lemma~1]{ckpps-facsd-16} for an alternate proof using the Perron-Frobenius theorem.)

Define a new weight vector $\mu$ by setting $\mu_d := \alpha_{\Tail(d)}\lambda_d$ for each dart $d$.  For each index~$i$, we immediately have $L^{\mu}_i P = \alpha^{\phantom{.}}_i L^\lambda_i P = \alpha^{\phantom{.}}_i H^\lambda_i = H^{\mu}_i$, where $P$ is the position matrix for~$\Gamma$, so $\mu$ is in fact a barycentric weight vector for~$\Gamma$.  
Finally, we observe that
\(
	(1,\dots,1) L^\mu = \alpha L^\lambda = (0,\dots,0)
\)
and
\(
	(1,\dots,1) H^\mu = \alpha H^\lambda = \alpha L^\lambda P = (0,\dots,0)P = (0,0),
\)
which imply that $\mu$ is morphable.
\end{proof}


\begin{theorem}
\label{T:torus-morph}
Given coordinate representations of two isotopic essentially 3-connected geodesic torus drawings $\Gamma_0$ and $\Gamma_1$, we can efficiently compute a morph from $\Gamma_0$ to $\Gamma_1$.  Specifically, after $O(n^{\omega/2})$ preprocessing time, we can compute any intermediate drawing during the morph in $O(n^{\omega/2})$ time.
\end{theorem}

\begin{proof}
Suppose $\Gamma_0$ and $\Gamma_1$ are convex drawings.  First, if necessary, we normalize the given coordinate representations so that their translation vectors agree, in $O(n)$ time~\cite[Theorem~A.1]{celp-hmgt-21}.  Then we find barycentric weight vectors $\lambda(0)$ and $\lambda(1)$ for $\Gamma_0$ and $\Gamma_1$, respectively, in $O(n)$ time, for example using Floater's mean-value coordinates~\cite{f-mvc-03,hf-mvcap-06}.  Following Lemma~\ref{L:can-morphabilify}, we derive morphable weights $\mu(0)$ and $\mu(1)$ from $\lambda(0)$ and $\lambda(1)$, respectively, in $O(n^{\omega/2})$ time.  Finally, given any real number $0 < t < 1$, we set $\mu(t) := (1-t)\mu(0) + t\mu(1)$ and solve the linear system $L^{\mu(t)} P^{(t)} = H^{\mu(t)}$ for the position matrix $P^{(t)}$ of an intermediate drawing $\Gamma^{\mu(t)}$; Lemmas~\ref{L:morphable-convex} and \ref{L:morphable-is-realizable} imply that this system is solvable.  The function $t \mapsto \Gamma^{\mu(t)}$ is a convexity-preserving morph between $\Gamma_0$ and $\Gamma_1$.

If the faces of $\Gamma_0$ or $\Gamma_1$ are not convex, we morph through an intermediate convex drawing, similarly to Chambers \etal~\cite[Theorem 8.1]{celp-hmgt-21}.  Let $\Gamma_*$ be the Floater drawing of~$G$ obtained by setting every dart weight to $1$.  Compute any triangulation~$T_0$ of $\Gamma_0$, and then triangulate the convex faces $\Gamma_*$ using the same diagonals, to obtain a triangulation~$T_*$ isotopic to $T_0$.  Assign weight~$0$ to the darts of the diagonals in $T_*\setminus \Gamma_*$ to obtain a barycentric weight vector $\mu_*$ for~$T_*$, which is symmetric and therefore morphable.  Derive morphable weights $\mu_0$ for~$T_0$ using mean-value coordinates \cite{f-mvc-03,hf-mvcap-06} and Lemma~\ref{L:can-morphabilify}.  Then we can morph from $T_0$ to $T_*$ by weight interpolation, using the weight vector $\mu(t) := (1-2t)\mu_0 + 2t\mu_*$ for any $0\le t\le 1/2$.  Ignoring the diagonal edges gives us a morph from $\Gamma_0$ to~$\Gamma_*$.  A symmetric procedure yields a morph from~$\Gamma_*$ to~$\Gamma_1$.
\end{proof}
%
%

\subsection{Deformation Space of Geodesic Torus Drawings}

Our formulation of morphable weights provides a straightforward solution to a conjecture of Connelly~\etal \cite{chhs-prlhe-83} about the deformation space of geodesic torus triangulations.  Bloch, Connelly, and Henderson \cite{bch-sslhc-84} proved that for any planar straight-line triangulation $\Gamma$ of a convex polygon $P$, the space of all planar straight-line triangulations of $P$ that are homeomorphic to~$\Gamma$ is contractible.  (Cairns’ morphing theorem \cite{c-dprc-44, c-idgc2-44} asserts only that this space is connected.)  Simpler proofs of this theorem were recently given by Cerf \cite{c-abcht-19} and Luo \cite{l-sgts-20}; in particular, Luo observed that the Bloch--Connelly--Henderson theorem follows immediately from Floater’s barycentric embedding theorem.

Connelly~\etal~\cite{chhs-prlhe-83} conjectured that every isotopy class of geodesic triangulations on any surface $S$ with constant curvature is homotopy-equivalent to the group $\emph{Isom}_0(S)$ of isometries of~$S$ that are homotopic to the identity.  In particular, $\emph{Isom}_0(S^2)$ is the rotation group $SO(3)$ and $\emph{Isom}_0(\Torus)$ is the translation group $S^1\times S^1$; for every other orientable surface $S$ without boundary, $\emph{Isom}_0(S)$ is trivial \cite{n-ihi-89}.  Very recently, Luo, Wu, and Zhu proved this conjecture for all surfaces of genus at least $2$ \cite{lwz-dsgtg-21} and for the flat torus \cite{lwz-dsgtf-21}; both proofs use nontrivial extensions of Floater’s theorem.  

Here we offer a simpler proof for torus graphs; in fact, we prove a more general result about convex drawings instead of just triangulations.

\begin{theorem}
For any convex drawing $\Gamma$ on the flat torus $\Torus$, the space of all convex drawings isotopic to $\Gamma$ is homotopy equivalent to $\Torus$.
\end{theorem}

\begin{proof}
Fix a convex drawing $\Gamma$ of a graph $G$ with $n$ vertices and $m$ edges; without loss of generality, assume some vertex $v$ is positioned at $(0,0)$.  Let $X = X(\Gamma)$ denote the space of all convex drawings of $G$ isotopic to~$\Gamma$, and let $X_0 = X_0(\Gamma)$ be the subspace of drawings in $X$ where vertex $v$ is positioned at $(0,0)$.  Every drawing in $X$ is a translation of a unique drawing in $X_0$, so $X = X_0 \times S^1 \times S^1$.  Thus, to prove the theorem, it suffices to prove that $X_0$ is contractible.

Call a weight vector $\lambda\in\Real_+^{2m}$ for $\Gamma$ \emph{normalized} if $\sum_d \lambda_d = 1$, where the sum is over all darts of~$\Gamma$.  Let $R = R(\Gamma)$ denote the set of of all realizable weight vectors for $T$, and let $\overline{M} = \overline{M}(\Gamma)$ denote the set of all normalized morphable weight vectors for $T$.

Lemma~\ref{L:morphable-convex} implies that $\overline{M}$ is convex and therefore contractible.  (Specifically, $\overline{M}$ is the interior of a $(2m-n-2)$-dimensional convex polytope in $\Real^{2m}$.)

Call two realizable weight vectors $\lambda,\lambda’\in R$ \emph{equivalent} if there is a scaling vector $\alpha\in\Real_+^n$ such that $\lambda’_d = \alpha_{\Tail(d)}\lambda_d$ for every dart $d$.  Because the Laplacian matrix $L^\lambda$ has rank $n-1$, Lemma~\ref{L:can-morphabilify} implies that every realizable weight $\lambda\in R$ is equivalent to a \emph{unique} normalized morphable weight $\mu\in\overline{M}$.  It follows that $R$ is homeomorphic to $\overline{M} \times \Real_+^n$ and therefore contractible.

Now we follow the proof of Theorem 1.4 in Luo \etal~\cite{lwz-dsgtg-21}.  Because every morphable weight is realizable, solving linear system \eqref{eq:ggt} gives us a continuous map $\Phi\colon R\to X_0$.  Floater’s mean-value weights \cite{f-mvc-03,hf-mvcap-06} give us a continuous map $\Psi\colon X_0 \to R$ such that $\Phi\circ\Psi$ is the identity map on $X_0$.  Because $R$ is contractible, the function $\Psi\circ\Phi$ is homotopic to the identity map on $R$.  We conclude that $X_0$ is homotopy equivalent to $R$ and therefore contractible.
\end{proof}
\section{Open Questions}

It is natural to ask whether our “best-of-both-worlds” planar morph can be extended to graphs on the flat torus.  In Appendix~\ref{A:no-torus-edgebyedge}, we prove a toroidal analog of Lemma~\ref{L:tutte-parallel-plane} for \emph{realizable} weight vectors; unfortunately, the main roadblock is that not all weight vectors are realizable.  In particular, given a realizable weight vector (morphable or not), it is not clear when changing the weights for a single edge results in another realizable weight vector.  

Several previous planar morphing algorithms~\cite{aabcd-hmpgd-17, kklss-cimpg-19, alflp-omcd-15, ddfpr-upm-20} rely on a certain convexifying procedure \cite{hn-cdhpg-10,k-cdhgl-21,kklss-cimpg-19,cgt-cdg23-96}, and are (potentially) faster than our algorithm via the implementation recently described by Klemz~\cite{k-cdhgl-21}.
It is an open question whether the procedure can be extended to geodesic torus graphs.

One can also ask if the result can be extended to surfaces of higher genus. The recent results of Luo \etal~\cite{lwz-dsgtg-21} imply that Floater and Gotsman's planar morphing algorithm~\cite{fg-mti-99} extends to geodesic triangulations on higher-genus surfaces of negative curvature; however, the existence of (any reasonable analog of) piecewise-linear morphs on such surfaces remains unknown.

\subsection*{Acknowledgments}

We thank Anna Lubiw for asking questions about Lemma 5.1 of Chambers \etal~\cite{celp-hmgt-21}, whose answers ultimately led to the discovery of Theorem~\ref{T:planar-morph}, and for other helpful feedback.
We also thank Yanwen Luo for making us aware of his recent work~\cite{l-sgts-20, lwz-dsgtg-21, lwz-dsgtf-21}.
Finally, we thank the anonymous reviewers for their comments and helpful suggestions for improvement.

\bibliographystyle{newuser-doi}
\bibliography{bib/topology,bib/algorithms,bib/geom}

\appendix

\section{Some Bad Examples}
\label{A:no-torus-fg}
\label{A:no-torus-sf}

Here we provide several concrete examples showing that barycentric methods for drawing and morphing planar graphs do not immediately generalize to graphs on the torus.

First, we give an infinite family of non-realizable positive weight vectors.  Let $G$ be \emph{any} graph on the torus, and consider the weight vector $\lambda$ that assigns weight $2$ to a single dart $d$ and weight~$1$ to every dart other than $d$.  (There is obviously nothing special about the values $1$ and $2$ here.)

\begin{lemma}
$\lambda$ is not a realizable weight vector for $G$. 
\end{lemma}

\begin{proof}
Let $u = \Tail(d)$ and $v = \Head(d)$.  For the sake of argument, suppose the linear system $L^\lambda P = H^\lambda$ has a solution; let $\Gamma^\lambda$ be the resulting drawing.  This system remains solvable if we remove row $u$ and arbitrarily fix $p_u$~\cite{sf-ppc2m-04}.  All dart weights in this truncated linear system are equal to~$1$, which implies that the drawing $\Gamma^\lambda$ is identical to the Tutte drawing $\Gamma^1$ for the all-$1$s weight vector.  Comparing the two linear systems, we conclude that $p_v - p_u + \tau_d = (0,0)$; that is, the edge of~$d$ has length zero in $\Gamma^\lambda = \Gamma^1$.  But this is impossible; every edge in a Tutte drawing has non-zero length \cite[Lemma B.5]{ggt-domam-06}.
\end{proof}

Next, we give an example of two realizable weight vectors for the same torus graph whose averages are not realizable.  Consider the toroidal drawings of $K_7$ shown in Fig.~\ref{F:bad-K7s}, which differ only in the position of vertex $2$.  We computed mean-value weights $\lambda$ and $\mu$ for these drawings, normalized so that the weights of all edges leaving each vertex sum to $1$~\cite{f-mvc-03,hf-mvcap-06}.  Routine calculations (which we implemented in Python) now imply that the average weight $(\lambda+\mu)/2$ is not realizable.

\begin{figure}[htb]
\centering
	\includegraphics[width=0.25\linewidth]{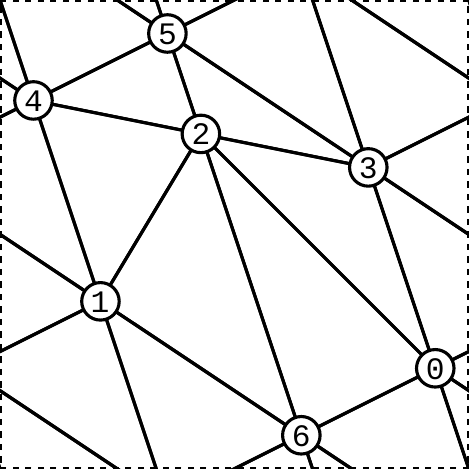}	
	\qquad\qquad
	\includegraphics[width=0.25\linewidth]{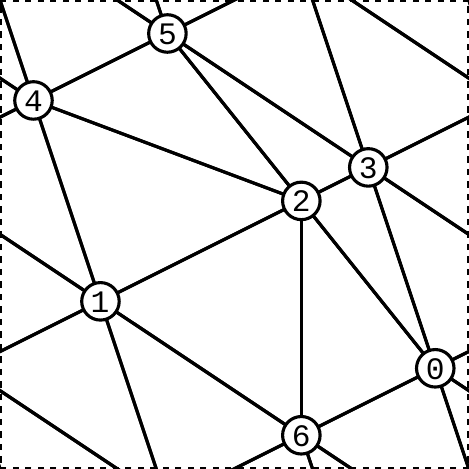}	
\caption{Isotopic drawings of $K_7$ whose normalized mean-value weights are not morphable.}
\label{F:bad-K7s}
\end{figure}

Finally, we consider Steiner and Fischer’s approach \cite{sf-ppc2m-04} of fixing a single vertex, which restores the Laplacian linear system to full rank.  The top row of Fig.~\ref{F:bad-sf-morph} shows two isotopic drawings of a $12\times 12$ toroidal grid, one with a single row of vertices shifted $1/2$ to the left, the other with a single column of vertices shifted $1/2$ downward.  Let $\lambda^\leftarrow$ and $\lambda^\downarrow$ respectively denote the normalized mean-value weights for these drawings~\cite{f-mvc-03,hf-mvcap-06}.  The bottom left image in Fig.~\ref{F:bad-sf-morph} shows the Steiner-Fischer drawing for the weight $\lambda = (2 \lambda^\leftarrow + \lambda^\downarrow)/3$, with the red edges indicating the fixed vertex.  This drawing is clearly not crossing-free; it also follows that the weight vector $\lambda$ is not realizable.

The bottom right of Fig.~\ref{F:bad-sf-morph}  shows the corresponding Floater drawing for the realizable weight vector $\mu = (2 \mu^\leftarrow + \mu^\downarrow)/3$, where $\mu^\leftarrow$ and $\mu^\downarrow$ are \emph{morphable} weights derived by rescaling $\lambda^\leftarrow$ and~$\lambda^\downarrow$, as described in Lemma \ref{L:can-morphabilify}.

\begin{figure}[h]\centering
	\includegraphics[width=0.4\linewidth]{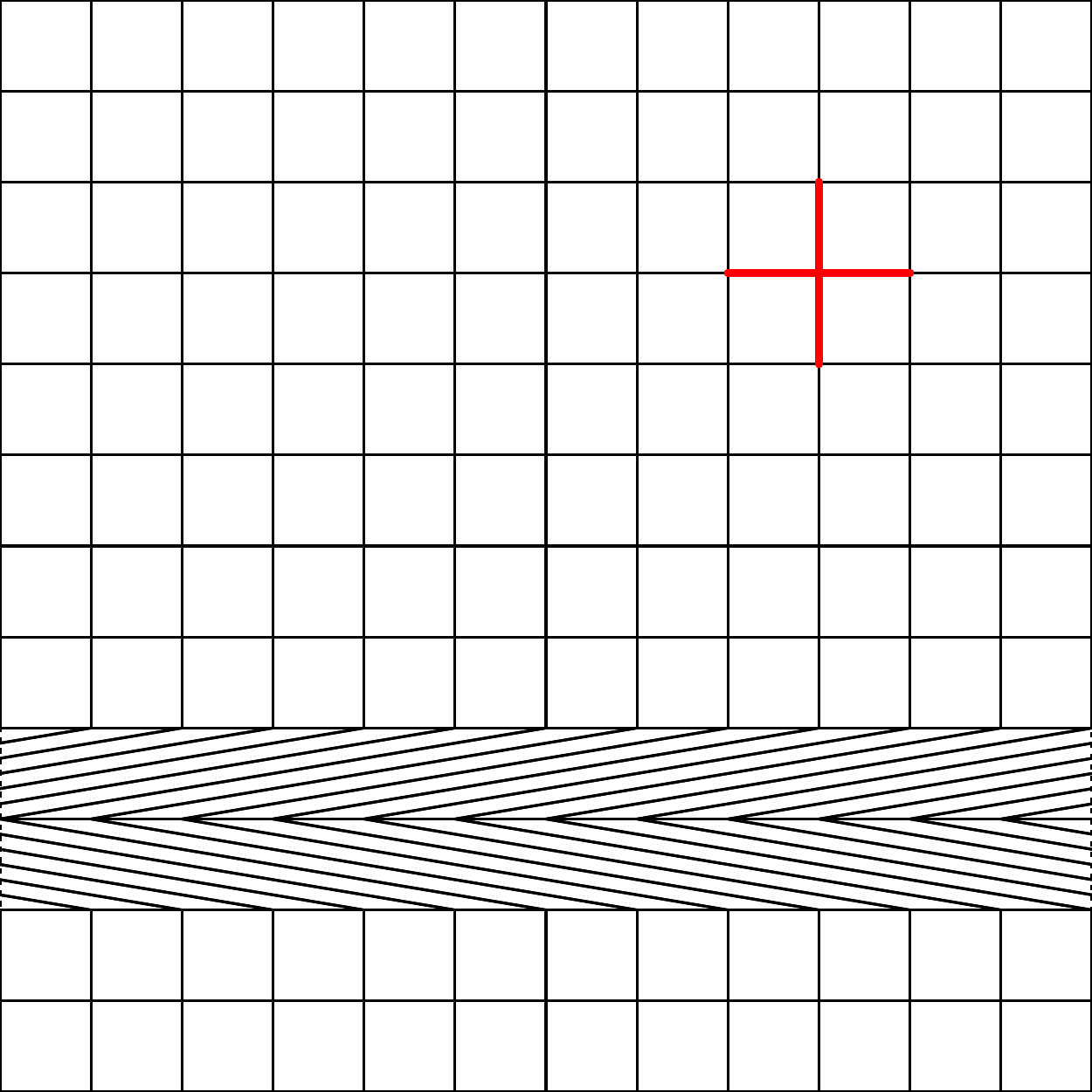}	
	\quad
	\includegraphics[width=0.4\linewidth]{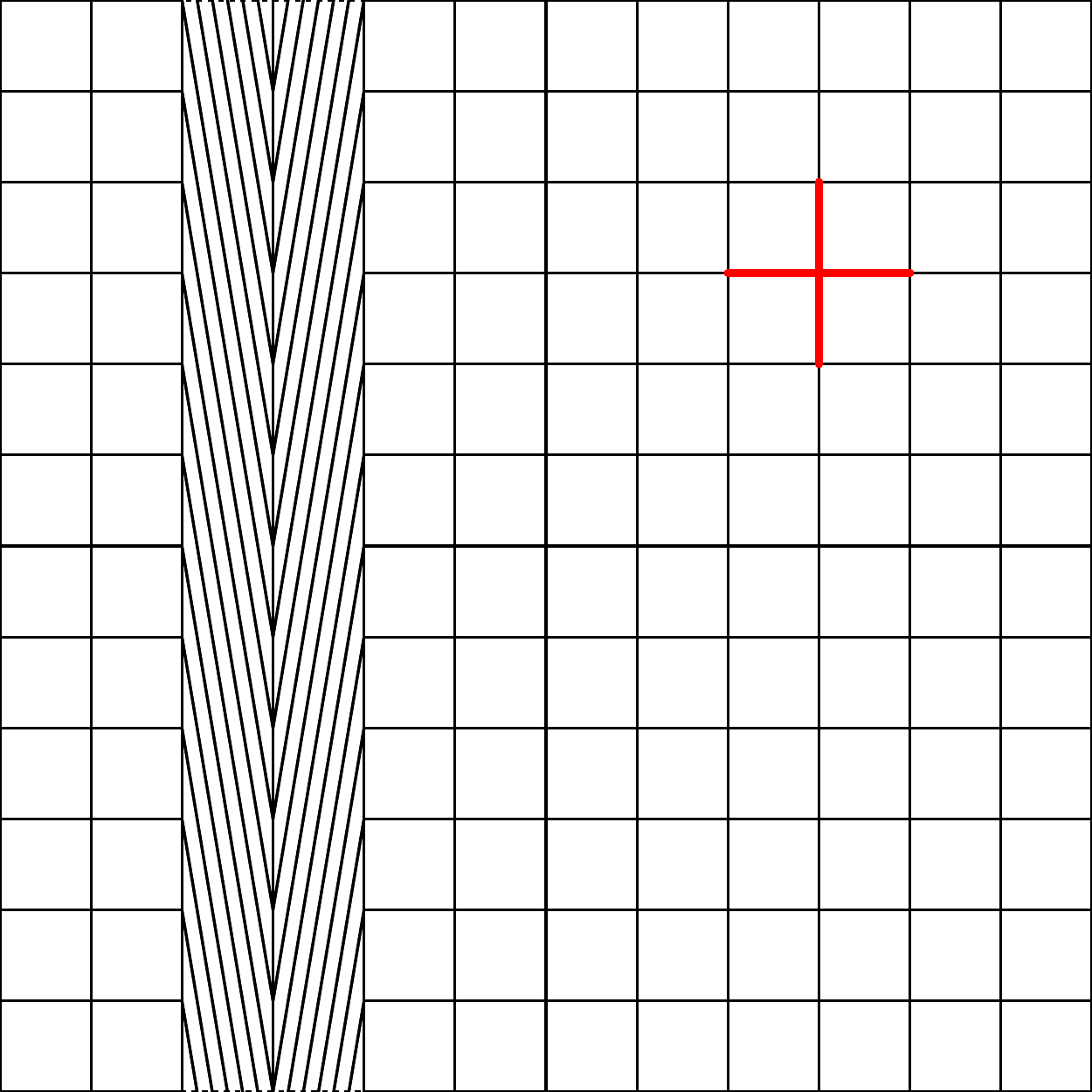}
	\\[3ex]
	\includegraphics[width=0.4\linewidth]{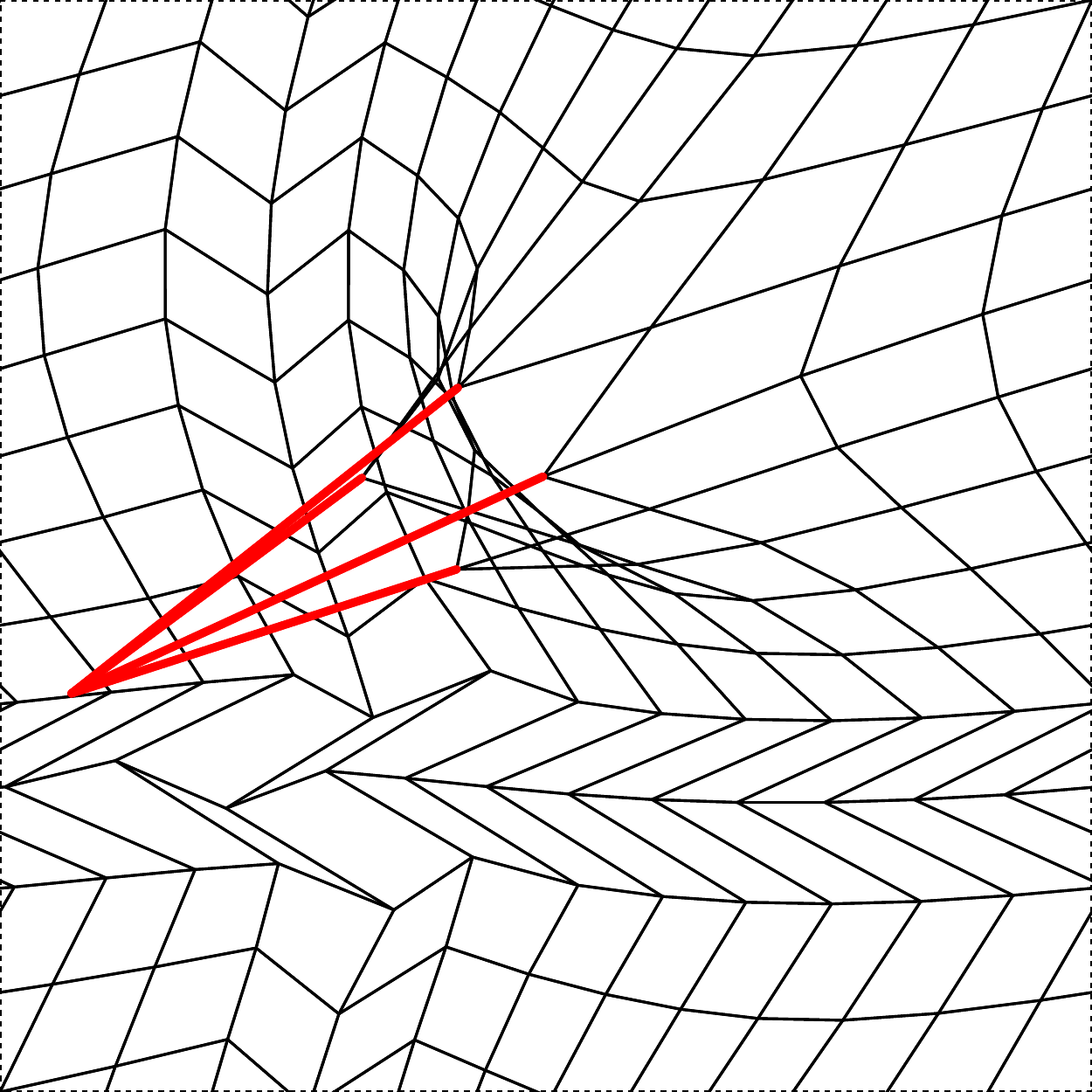}	
	\quad
	\includegraphics[width=0.4\linewidth]{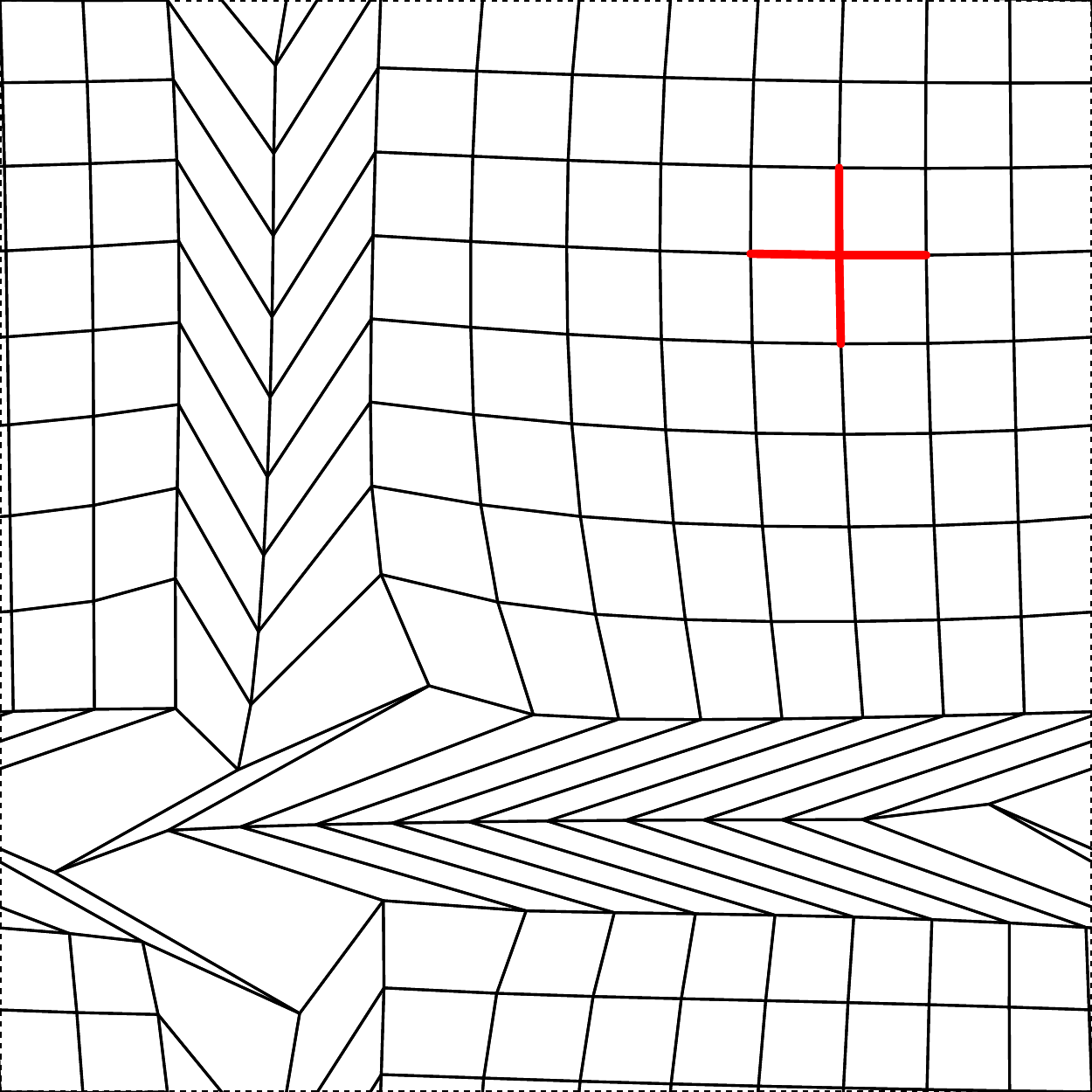}
\caption{A bad example of fixed-vertex weight interpolation; see the text for explanation.}
\label{F:bad-sf-morph}
\end{figure}

Steiner and Fischer claim \cite[Section 2.2.1]{sf-ppc2m-04} that their drawings have no “foldovers'' except possibly at the fixed vertex and its neighbors.  Here a ``foldover'' is a vertex in the drawing whose incident faces overlap (where formally, faces are determined by the reference drawing used to define the translation vectors).  Close examination of Fig.~\ref{F:bad-sf-morph} shows that this claim is incorrect.

\section{Trying to Morph Torus Graphs Edge by Edge}
\label{A:no-torus-edgebyedge}

The following lemma generalizes Lemma~\ref{L:tutte-parallel-plane} to the toroidal setting; it also generalizes Lemma 5.1 of Chambers \etal~\cite{celp-hmgt-21} to work for \emph{asymmetric} weights.

\begin{lemma}
\label{L:tutte-parallel-torus}
Let $\lambda$ and $\mu$ be arbitrary realizable weight vectors such that $\lambda_{d} \ne \mu_{d}$ or $\lambda_{\Rev(d)} \ne \mu_{\Rev(d)}$ for some dart $d$, and $\lambda_{d'} = \mu_{d'}$ for all darts $d' \notin \Set{d,\Rev(d)}$.  For every vertex $w$, the vector $p^\mu_w - p^\lambda_w$ is parallel to the drawing of $d$ in $\Tutte{\Gamma}{\lambda}$.
\end{lemma}

\begin{proof}
Suppose $d$ has tail $u$ and head $v$.  By the definition of crossing vectors, dart $d$ appears in $\Tutte{\Gamma}{\lambda}$ as the projection of a dart in the universal cover from $p^\lambda_u$ to $p^\lambda_v + \tau_{d}$.  Fix a non-zero vector $\sigma \in \Real^2$ orthogonal to the vector $p^\lambda_v - p^\lambda_u + \tau_{d}$ and thus orthogonal to darts $\Set{d,\Rev(d)}$ in~$\Tutte{\Gamma}{\lambda}$.  For each vertex $i$, let $z^\lambda_i = p^\lambda_i \cdot \sigma$ and $z^\mu_i= p^\mu_i \cdot \sigma$, and for each dart $d'$, let $\chi_{d'} = \tau_{d'} \cdot \sigma$.  Our choice of $\sigma$ implies that $\chi_{d} = z^\lambda_u - z^\lambda_v$.  We need to prove that $z^\lambda_i = z^\mu_i$ for every vertex $i$.

Let $X^\lambda = H^\lambda \cdot \sigma$ and $X^\mu = H^\mu \cdot \sigma$. The real vector $Z^\lambda = (z^\lambda_i)^{\phantom{.}}_i$ is a solution to the linear system $L^\lambda Z = X^\lambda$; in fact, $Z^\lambda$ is the \emph{unique} solution such that $z^\lambda_n = 0$. Similarly, $Z^\mu = (z^\mu_i)^{\phantom{.}}_i$ is the unique solution to an analogous equation $L^{\mu} Z = X^\mu$ with $z^\mu_n = 0$.
We will prove that $L^{\mu} Z^\lambda = X^{\mu}$, so that in fact $Z^\lambda = Z^\mu$.

Let $\delta = \mu_{d} - \lambda_{d}$ and $\e = \mu_{\Rev(d)} - \lambda_{\Rev(d)}$.  The matrices $L^{\lambda}$ and $L^{\mu}$ differ in only four locations:
\[
	L^{\mu}_{ij} - L^{\lambda}_{ij} = \begin{cases}
		\delta & \text{if $(i,j) = (u,u)$} \\
		-\delta & \text{if $(i,j) = (u,v)$} \\
		-\e & \text{if $(i,j) = (v,u)$} \\
		\e & \text{if $(i,j) = (v,v)$} \\
		0 & \text{otherwise}
	\end{cases}
\]
More concisely, we have $L^{\mu} = L^\lambda + (\delta\Unit_u - \e\Unit_v)\,(\Unit_u - \Unit_v)^T$.  Similar calculations imply $H^{\mu} = H^\lambda + \tau_{d}(\delta\Unit_u - \e\Unit_v)$ and therefore $X^{\mu} = X^\lambda + \chi_{d}(\delta\Unit_u - \e\Unit_v)$.  It follows that 
\begin{align*}
	L^{\mu} Z^\lambda
	&=
	L^\lambda Z^\lambda + (\delta\Unit_u - \e\Unit_v)\,(\Unit_u - \Unit_v)^T\, Z^\lambda 
\\	&=
	X^\lambda + (\delta\Unit_u - \e\Unit_v)\,(z^\lambda_u - z^\lambda_v)
\\	&=
	X^\lambda + (\delta\Unit_u - \e\Unit_v)\,\chi_{d}
\\	&=
	X^{\mu},
\end{align*}
completing the proof.
\end{proof}

Colin de Verdière~\cite{c-crgtd-91} showed that all \emph{symmetric} positive weights are realizable on the flat torus (see also \cite{d-eppgc-04,ggt-domam-06,l-dafe-04,hs-seshm-15}); Chambers \etal~\cite{celp-hmgt-21} exploit this observation in their morphing algorithm.  In the asymmetric case, however, it is unclear when changing the weights for a single edge in a realizable weight vector results in another realizable weight vector. We obtain the following partial result, via the analyses of Lemmas~\ref{L:can-morphabilify} and \ref{L:tutte-parallel-torus}.

\begin{lemma}
\label{L:torus-edge-tweak}
Let $\lambda$ be a realizable weight vector, and let $\alpha$ be a positive row vector such that $\alpha L^\lambda = (0,\dots,0)$ and $\alpha H^\lambda = (0,0)$.  Let $\mu$ be another positive weight vector such that $\lambda_{d} \ne \mu_{d}$ or $\lambda_{\Rev(d)} \ne \mu_{\Rev(d)}$  for some dart $d$, and $\lambda_{d'} = \mu_{d'}$ for all darts $d' \notin \Set{d,\Rev(d)}$.  Set $\delta := \mu_{d} - \lambda_{d}$ and $\e := \mu_{\Rev(d)} - \lambda_{\Rev(d)}$.  If $\delta\alpha_{\Tail(d)} = \e\alpha_{\Head(d)}$, then $\mu$ is realizable.
\end{lemma}

\begin{proof}
Suppose $d$ has tail $u$ and head $v$. The analysis of Lemma~\ref{L:tutte-parallel-torus} gives us
\begin{align*}
	L^{\mu} &= L^\lambda + (\delta\Unit_u - \e\Unit_v)\,(\Unit_u - \Unit_v)^T \\
	H^{\mu} &= H^\lambda + \tau_{d}(\delta\Unit_u - \e\Unit_v)
\end{align*}
Because $\alpha L^\lambda = (0,\dots,0)$ and $\alpha H^\lambda = (0,0)$, we immediately have $\alpha L^{\mu} = (\delta\alpha_u - \e\alpha_v)(\Unit_u - \Unit_v)^T$ and $\alpha H^{\mu} = x_{d}(\delta\alpha_u - \e\alpha_v)$.

If $\delta\alpha_u = \e\alpha_v$, then $\alpha L^{\mu} = (0,\dots,0)$ and $\alpha X^{\mu} = (0,0)$.  It follows that a suitable scaling of $\mu$ is morphable, and therefore realizable, which implies that $\mu$ itself is also realizable.
\end{proof}

In particular, when $\lambda$ is symmetric, then we can choose $\alpha$ to be the all-$1$s vector, which implies that any weight $\mu$ satisfying the conditions of Lemma \ref{L:torus-edge-tweak} is also symmetric.

\end{document}